\documentclass{article}
\usepackage[utf8]{inputenc}
\usepackage{amsmath}
\usepackage{amssymb}
\usepackage{amsthm}

\newtheorem{theorem}{Theorem}[section]
\newtheorem{definition}[theorem]{Definition}
\newtheorem{claim}[theorem]{Claim}
\newtheorem{remark}[theorem]{Remark}

\DeclareMathOperator{\cave}{c_{ave}}

\title{On the Geometrical and Kinematical Foundations of the Symmetric Relativity Model: Lorentz Transformation and Time Dilation}
\author{M. Mateljevi\'c}
\date{June 29, 2026}

\begin{document}

\maketitle

\begin{abstract}
We examine the kinematic foundations of relativity by considering two inertial frames, $S$ and $S'$, in a standard configuration, where $S'$ moves along the common spatial $x$-axis at a constant velocity $v$. By relaxing Einstein’s second postulate regarding the universal invariance of the one-way speed of light, we adopt an operational framework grounded strictly on the directly observable two-way (round-trip) speed of light, evaluated alongside the principles of spacetime homogeneity, linearity, and reciprocity.

Within this setting, we demonstrate that:
(i) the classical Lorentz transformation is recovered exactly along the coordinate axes under a generalized two-way synchronization scheme without requiring Einstein's second postulate;
(ii) Langevin’s light-clock argument fundamentally implies that the longitudinal scale factor $b$ matches the standard Lorentz factor $\gamma$; and
(iii) transverse lengths remain strictly invariant ($b_z = 1$).

Crucially, to resolve the ontic paradox of multiple co-existing wavefront centers, we introduce a kinematically symmetric model relative to an absolute cosmological rest frame (or geometric anchor) $K$, wherein $S$ and $S'$ move with equal and opposite velocities ($u$ and $-u$, respectively). Within this generalized framework, allowing for anisotropic one-way light propagation via Reichenbach-type parameters ($\varepsilon$ or $\kappa$) yields a consistent, linear velocity-addition law, a generalized Doppler effect, and a flat but oblique spacetime metric. Finally, we prove that all round-trip observables remain strictly invariant under synchronization gauge transformations (reflecting the Tangherlini--Edwards perspective) and demonstrate that the resulting non-diagonal metric structure is fully consistent withthe Sagnac effect over closed spatial loops.
\end{abstract}

\section{Foundational Framework and Cosmological Context}

By relaxing Einstein's second postulate of light-speed invariance, this framework is established strictly upon the observable two-way (round-trip) speed of light. This operational approach allows a rigorous exploration of mathematical and kinematical dimensions, focusing specifically on spacetime homogeneity, linearity, and reciprocity between inertial frames. Concurrently, it reassesses classical relativistic transformations and the invariance of physical laws under relaxed axiomatic constraints.

According to the orthodox Einsteinian interpretation, absolute space and time do not exist. To bridge this view with contemporary astrophysics—which provides a critical distinction between fundamental geometric laws and specific cosmological solutions—we distinguish between two underlying physical and mathematical paradigms:
\begin{enumerate}
    \item[\textit{(a)}] \textbf{The Absolute Background Paradigm:} The postulation of an immobile spatial background. This can be operationalized either empirically via the Cosmic Microwave Background (CMB) rest frame acting as a cosmic fluid anchor, or treated strictly as an idealized stationary point within a formal mathematical model.
    \item[\textit{(b)}] \textbf{The Relational Paradigm:} The complete absence of an absolute space, where physical laws possess no inherent reference point, characterized by Einstein's symmetric framework of mutually "stationary" observers.
\end{enumerate}

In the relational case \textit{(b)}, modern cosmology still reveals at least a local and operational need for measurements that transcend purely relative configurations. While we cannot definitively state whether a deeper ontological stratum of reality dictates an absolute spacetime fabric, we can rigorously analyze the kinematic consequences of this paradigm by introducing hypothesis \textbf{(H1)}.

\subsection*{Model Assumptions and Kinematic Framework}
\noindent\textit{\textbf{Model Note:} In this section, we strictly adopt the relational configuration but relax the standard light-speed invariance by operationalizing the system under the following kinematic symmetry hypothesis.}

\vspace{0.5em}
\begin{center}
\begin{minipage}{0.9\textwidth}
\small\itshape
\textbf{(H1)} Einstein's kinematically symmetric framework of "stationary" observers holds between inertial frames, formulated \textbf{without} the second postulate.
\end{minipage}
\end{center}
\vspace{0.5em}

\noindent \textbf{Operational Context of Hypothesis \textit{(H1)}:}
Throughout the ensuing analysis of Langevin's light clock, we explicitly assume \textit{(H1)}. While paradigm \textit{(a)} provides the conceptual motivation for an absolute geometric anchor, the mathematical derivations in this section rest strictly on the kinematic symmetry of \textit{(H1)}. This allows us to isolate how spatial and temporal scaling emerge purely from the homogeneity of space and time, independently of the constancy of the speed of light.

\subsection{The Mathematical Distinction: Laws vs. Solutions}
Modern theoretical physics reconciles the relational and absolute viewpoints by bifurcating the description of the universe into two distinct components:
\begin{itemize}
    \item \textbf{The Fundamental Laws:} The underlying dynamical equations of physics (such as Einstein's field equations) are strictly Lorentz-invariant. They contain no intrinsic immobile points, no absolute background, and no preferred physical velocity.
    \item \textbf{The Cosmological Solution:} When these fundamental equations are solved for our specific universe, the initial and boundary conditions must be factored in. Because the Big Bang distributed matter and radiation almost isotropically on a large scale, the resulting solution possesses a unique, mathematically preferred rest frame.
\end{itemize}
Unifying these perspectives, an immobile point or a stationary background exists as a rigorous mathematical property of the cosmological \textit{solution}, rather than an inherent property of the underlying physical \textit{law}.

\subsection{The Cosmic Microwave Background as a Physical Proxy}
Instead of postulating a metaphysical, undetectable absolute space, modern cosmology operationalizes the stationary background through the comoving coordinate system. Within this framework, the "immobile medium" is represented by the "cosmic fluid"---the average distribution of all matter and radiation in the universe, anchored empirically by the Cosmic Microwave Background (CMB).

The universe is permeated by a sea of relic photons left over from the early universe. For an observer strictly at rest relative to this radiation field, the CMB temperature appears perfectly isotropic at approximately $2.725\text{ K}$ in all directions. Conversely, any observer moving relative to this cosmic fluid encounters a temperature dipole due to Doppler shifts: a blueshift (warming effect) in the direction of motion and a redshift (cooling effect) behind. This anisotropy allows the calculation of our peculiar velocity relative to the cosmic rest frame. If an observer is stationary relative to this fluid, their clock ticks at the maximum possible rate for the age of the universe, defining Cosmic Time. Any local deviation from this frame is measurable as a specific velocity, rendering the framework mathematically equivalent to a localized absolute space.

\subsection{Mathematical Formalization of the Stationary Comoving Source}
To operationalize the concept of a "stationary point" within the cosmic fluid, we employ the Friedmann-Lemaître-Robertson-Walker (FLRW) metric. Let a point $A$ represent a localized light source embedded directly within the cosmic substrate. By definition, a source that is "at rest" relative to this medium possesses zero peculiar velocity ($v_{\text{pec}} = 0$). Consequently, its spatial position remains fixed at the origin ($r = 0$) of the comoving coordinate system, even as the global hypersurface undergoes cosmological expansion.

The invariant spacetime interval $ds^2$ for a spatially flat universe ($k=0$) is expressed as:
\begin{equation}
ds^2 = -c^2 dt^2 + a^2(t) \left[ dr^2 + r^2 (d\theta^2 + \sin^2\theta d\phi^2) \right]
\end{equation}
where $a(t)$ is the dimensionless scale factor quantifying the expansion of the universe, and $t$ represents the Cosmic Time ticked by a comoving observer at rest relative to the CMB.

When a light signal is emitted isotropically from point $A$ at an initial cosmic time $t_0$, the propagation of the wave front follows a null geodesic ($ds^2 = 0$). For radial propagation ($d\theta = d\phi = 0$), the differential relation simplifies to:
\begin{equation}
c \, dt = a(t) \, dr.
\end{equation}

By separating variables, the comoving coordinate distance $r$ traversed by the photon field up to a subsequent cosmic time $t$ is determined through integration:
\begin{equation}
r = \int_{t_0}^{t} \frac{c \, dt'}{a(t')}.
\end{equation}

This integration demonstrates that the comoving coordinate $r$ depends strictly on the global temporal parameters $t_0$ and $t$. Because the coordinates of point $A$ are time-invariant ($r_A = 0$), the geometric center of the expanding light sphere remains permanently anchored at the source.

While the physical distance to the wave front scales dynamically as $D(t) = a(t)r$, the emission point $A$ experiences no physical displacement through the coordinate mesh. The cosmic fluid thus serves as a concrete, non-metaphysical medium: light propagates at velocity $c$ relative to the local comoving space, while the isotropic emitter $A$ remains perfectly stationary within the geometric fabric of the universe.
\subsection{Axiomatic Postulates vs. Deeper Physical Mechanisms}
It is crucial to note that Einstein elevated the constancy of the one-way speed of light to an axiomatic postulate rather than deriving it from a deeper physical mechanism. In a strict absolute background model (often generalized within the Neo-Lorentzian framework), the speed of light is constant relative to the stationary medium, while moving observers experience physical length contraction and time dilation as they move through this background.


Mathematically, both Einstein's symmetric approach and the absolute background model yield identical experimental predictions.

\section{Setup and Two-Way Synchronization on the Axis}


\subsection{Setup and Two-Way Synchronization on the Axis}

We consider inertial frames $S$ and $S'$ in a standard configuration: parallel spatial axes, with $S'$ moving at speed $v$ along the $+x$ axis of $S$, and origins coincident at $t = t' = 0$.

\begin{definition}[Two-way synchronization along $x$]
Let $S$ be an inertial frame with origin $O$ and unit vector $e_1 = (1, 0, 0)$. Denote by $c_1 = c_{\text{ave}}(e_1)$ the two-way (round-trip) speed of light along $e_1$. For a point $M$ on the $x$-axis with $d = |OM|$, a light pulse emitted from $O$ at $t = 0$ starts the clock at $M$ at time $d/c_1$ upon its arrival. We say clocks on the $x$-axis are synchronized via the two-way speed in the $x$ direction.
\end{definition}

Let $c_1(M)$ be the two-way (round-trip) speed of light from $O$ to $M$ and back from $M$ to $O$. By space homogeneity, $c_1(M)$ is independent of $M$, allowing us to define a global constant $c_1(S) \equiv c_1(M)$ for the frame $S$. Using symmetry and the Principle of Relativity, we can prove the following:

\begin{claim}
$c_1(S) = c_1(S')$.
\end{claim}

\begin{proof}
Let $S$ and $S'$ be two inertial frames. By definition, the two-way speed of light $c_1(S)$ in $S$ is measured via a co-linear round-trip experiment. A light pulse is emitted from $O$ at $t = 0$, reflects off a stationary mirror at $M$ (where $d = |OM|$), and returns to $O$ at time $t = T$. The speed is given by:
\[
c_1(S) = \frac{2d}{T}.
\]
By the operational principle of relativity, an identical experimental setup must be definable within $S'$. Let an observer at the origin $O'$ emit a light pulse along the $+x'$ axis at $t' = 0$ towards a mirror stationary in $S'$ at $M'$, at a proper distance $d' = |O'M'|$. The pulse returns to $O'$ at proper time $t' = T'$, yielding the two-way speed:
\[
c_1(S') = \frac{2d'}{T'}.
\]
According to the Principle of Relativity, formulated here by relaxing Einstein's second postulate regarding the invariance of the one-way speed of light, all inertial frames remain physically equivalent. The laws of physics, including the propagation of light and the operational definition of space and time intervals, must possess the exact same functional form in all inertial systems. No intrinsic physical property or absolute velocity can distinguish $S$ from $S'$.

Consequently, if the two experiments are geometrically identical in their respective rest frames such that the proper distances are equal ($d = d'$), the structural symmetry demands that the elapsed proper round-trip times must also be equal ($T = T'$). Dividing the identical proper distances by the identical proper times immediately implies $c_1(S) = c_1(S')$.
\end{proof}

\begin{remark}
It is crucial to emphasize that this proof \emph{does not} rely on Einstein's Second Postulate (the universal constancy of the one-way speed of light). Instead, $c_1$ is strictly defined as an operational, two-way (round-trip) average speed. The identity $c_1(S) = c_1(S')$ is a direct consequence of the Principle of Relativity and space homogeneity alone.

This separates convention from physics: the one-way speed of light requires synchronization of clocks (which is a matter of convention, e.g., Reichenbach's $\varepsilon$), while the two-way speed of light is measurable on a single clock and represents a pure physical fact. This raises the level of rigor, demonstrating that the theory is built on general geometric and physical symmetries rather than assumptions specific only to electromagnetic waves.
\end{remark}

Set $\gamma^0(u)=\frac{1}{\sqrt{1-u^2}}$ and $\gamma_1 = \gamma_1(v) = \gamma^0\left(\frac{v}{c_1}\right)$.

\begin{theorem}[Axis-restricted Lorentz form]\label{thm:axisLT}
Suppose (H1) paradigm. Under the above synchronization on $x$ (and analogously on $x'$), the standard Lorentz form relating $(x,t)$ and $(x',t')$ holds exactly for events restricted to the $x$- and $x'$-axes with $\gamma_1(v)$.
\end{theorem}
Only the two-way speed $c_1 = c_{\text{ave}}(e_1)$ enters the derivation for these axis-restricted transformations. The full proof is provided in Section \ref{Lange}.

\begin{remark}[Asymmetric vs. Symmetric Contraction Paradigm]\label{rem:symmetry_contrast}
A profound conceptual distinction emerges when comparing these two structural frameworks:
\begin{itemize}
    \item In the \textbf{Absolute Background Paradigm (a)}, the length contraction is fundamentally \textbf{asymmetric}. Because the base frame $K$ acts as a privileged, immobile geometric anchor, the structural contraction is unidirectional relative to this background, meaning that observers in $K$ and $S$ do not share a mathematically reciprocal view of physical dimensions.
    \item Conversely, under the \textbf{Kinematic Symmetry Hypothesis (H1)}, the contraction becomes strictly \textbf{symmetric} (reciprocal). Due to the pure relational nature of the framework, if an observer in $S$ measures a longitudinal contraction of a moving rod in $S'$ via the factor $\gamma_1$, an identical operational experiment conducted in $S'$ looking back at $S$ yields the exact same contraction factor.
\end{itemize}
This highlights that while the mathematical robustness of the longitudinal contraction formula survives even under transverse anisotropy $c_n \neq c_1$ in Paradigm (a), the underlying physical interpretation transforms from an asymmetric background effect into a beautiful, purely geometric reciprocity under Hypothesis (H1).
\end{remark}

\begin{remark}[Isotropy of the two-way speed]
Null results of Michelson--Morley--type experiments support $c_{\text{ave}}(e_1)=c_{\text{ave}}(e_3)=c$ while leaving the one-way speeds convention-dependent (Einstein/Edwards/Reichenbach/Tangherlini synchronizations).
\end{remark}

Let $n$ be a spatial unit vector orthogonal to the $x$-axis in $S$, and let $e_1 = (1, 0, 0)$ denote the unit vector along the $x$-axis. We define $c_n \equiv c_{\text{ave}}(n, S)$ as the transverse two-way speed of light and $c_1 \equiv c_{\text{ave}}(e_1, S)$ as the longitudinal two-way speed of light. Within the framework of the Absolute Background Paradigm~\textit{(a)}, we admit the general anisotropy condition where $c_n \neq c_1$ inside the moving frame.

\begin{theorem}[Robustness of contraction under anisotropy]\label{thm:anisotropy_contraction}
Let $S$ be an inertial frame moving with a constant velocity $v$ relative to the standard absolute inertial frame $K$ along the $x$-axis. If the operational two-way speed of light in $S$ is direction-dependent such that $c_n = c_{\text{ave}}(n, S) \neq c_{\text{ave}}(e_1, S) = c_1$, the longitudinal length contraction of a rod $R$ measured in $K$ remains strictly invariant and satisfies
\begin{equation}\label{eq:robust_contraction}
L_S = \gamma L_S^*,
\end{equation}
where $L_S$ is the proper length of $R$ in $S$, $L_S^*$ is the measure of $R$ from $K$, and $\gamma = \left(1 - \frac{v^2}{c^2}\right)^{-1/2}$ is defined via the isotropic speed of light $c$ in the base frame $K$.
\end{theorem}

\begin{proof}
In the transverse direction, the relation between the time intervals $T$ (measured in $K$) and $t$ (measured in $S$) for a light signal is given by:
\begin{equation}\label{eq:time_relation_proof}
T = \frac{c_n}{c}\gamma t.
\end{equation}
The round-trip time interval $t$ in $S$ for a transverse rod of proper length $L_S$ is operationally defined by $2L_S = c_n t$. Substituting $t = \frac{2L_S}{c_n}$ into \eqref{eq:time_relation_proof}, the anisotropic speed factor $c_n$ cancels out identically, yielding:
\begin{equation}\label{eq:transverse_T}
T = \frac{2L_S \gamma}{c}.
\end{equation}
In the longitudinal direction (along the $x$-axis), the observer in $K$ measures the contracted length of the rod as $L_S^*$. The round-trip time interval evaluated from the perspective of $K$ (where the speed of light is isotropic and equals $c$) is:
\begin{equation}\label{eq:lang2a_proof}
T = \left(\frac{2L_S^*}{c}\right)\gamma^2.
\end{equation}
Equating \eqref{eq:transverse_T} and \eqref{eq:lang2a_proof} gives:
\begin{equation}
\frac{2L_S \gamma}{c} = \left(\frac{2L_S^*}{c}\right)\gamma^2.
\end{equation}
Simplifying this identity immediately yields $L_S = \gamma L_S^*$, completing the proof.
\end{proof}

Using the proof of Theorem~\ref{thm:anisotropy_contraction} and the associated structural properties, we establish the following result:

\begin{theorem}\label{thm:inverse_transformation}
Suppose that $K$ is the absolute frame with coordinates $(X, T)$ and $S$ is the moving frame with coordinates $(x, t)$. The inverse spatial transformation and the transverse time relation are given by:
\begin{equation}\label{eq:transf_rules}
x = \gamma (X - vT), \quad \text{and} \quad T = \frac{c_n}{c}\gamma t.
\end{equation}
Information regarding a rod stationary in $K$ follows: let a rod of proper length $L$ be stationary in the absolute frame $K$. If its length is operationally measured via simultaneous observations in the moving frame $S$ as $L'$, then it satisfies:
\begin{equation}
L' = \gamma L.
\end{equation}
\end{theorem}

\begin{remark}[Algebraic Proof of Asymmetric Contraction]\label{rem:algebraic_asymmetric}
To explicitly demonstrate the asymmetric nature of the transformation under the Absolute Background Paradigm~\textit{(a)}, we examine the coordinate relations between the absolute frame $K$ and the moving frame $S$ given in \eqref{eq:transf_rules}.

By defining the effective anisotropy scale factor as $a = \frac{c_n}{c}\gamma$, the temporal relation simplifies to $T = a t$. Substituting this directly into the spatial transformation allows us to express the absolute coordinate $X$ in terms of the moving coordinates:
\begin{equation}
X = \frac{x}{\gamma} + v T = \frac{x}{\gamma} + v a t.
\end{equation}
Now, let us consider a rod stationary in the absolute frame $K$, whose length is operationally measured via simultaneous observations in the moving frame $S$. Imposing the simultaneity condition in $S$ such that $t_1 = t_2$ (implying $\Delta t = 0$), the spatial displacement in the absolute background simplifies to:
\begin{equation}
X_2 - X_1 = \frac{x_2 - x_1}{\gamma} \quad \implies \quad \Delta X = \frac{\Delta x}{\gamma}.
\end{equation}
Identifying $\Delta X = L$ as the proper length in $K$ and $\Delta x = L'$ as the measured length in $S$, this identity algebraically confirms that when simultaneity is defined relative to the moving frame $S$, the longitudinal structural relation depends strictly on the baseline factor $\gamma$. The directional anisotropy factor $a$ completely decouples from the spatial measurement, beautifully illustrating the asymmetric structural fabric of paradigm~\textit{(a)}.
\end{remark}


\section{Spatial and Temporal Scaling via Langevin's Light Clock}\label{Lange}

We begin by distinguishing between two fundamental physical and mathematical paradigms regarding the underlying structure of spacetime:
\begin{enumerate}
    \item[\textit{(a)}] \textbf{The Absolute Background Paradigm:} The postulation of an immobile background, operationalized either empirically via the Cosmic Microwave Background (CMB) rest frame acting as a cosmic fluid anchor, or treated strictly as an idealized stationary point within a mathematical model.
    \item[\textit{(b)}] \textbf{The Relational Paradigm:} The complete absence of any absolute space, where physical laws possess no inherent reference point.
\end{enumerate}

\subsection*{Model Assumptions and Kinematic Framework}
\noindent\textit{\textbf{Model Note:} In this section, we strictly adopt the relational configuration but relax the standard light-speed invariance by operationalizing the system under the following kinematic symmetry hypothesis.}

\vspace{0.5em}
\begin{center}
\begin{minipage}{0.9\textwidth}
\small\itshape
\textbf{(H1)} Einstein's kinematically symmetric framework of 'stationary' observers holds between inertial frames, formulated \textbf{without} the second postulate.
\end{minipage}
\end{center}
\vspace{0.5em}

\noindent \textbf{Operational Context of Hypothesis \textit{(H1)}:}
Throughout the ensuing analysis of Langevin's light clock, we explicitly assume \textit{(H1)}. While paradigm \textit{(a)} provides the conceptual motivation for an absolute geometric anchor, the mathematical derivations in this section rest strictly on the kinematic symmetry of \textit{(H1)}. This allows us to isolate how spatial and temporal scaling emerge purely from the homogeneity of space and time, independently of the constancy of the speed of light.

\begin{theorem}\label{ThSym}
Assume paradigm \textit{(H1)}. Consider linear and homogeneous coordinate transformations between the inertial frames $S$ and $S'$ in a standard configuration, where $S'$ moves with a constant relative velocity $v$ along the $x$-axis of $S$. Let $R'$ be a rod at rest in $S'$ with proper length $L'$. If an observer in $S$ measures the length of $R'$ as $L$, then $L' =\gamma L$.
\end{theorem}
Here, we suppose that the speed of light in vacuum is $c$.

\begin{proof}
By homogeneity, we set a temporal scale factor $a$ tracking the spatial origin of the moving frame ($x' = 0$) and a spatial scale factor $b$ relating the lengths:
\begin{equation}\label{eq:time_scale}
t = at' \quad (\text{for } x' = 0)
\end{equation}
\begin{equation}\label{eq:length_scale}
L' = bL
\end{equation}
From spatial reciprocity and the equivalence of inertial frames, the forward and inverse coordinate relations along the axis of motion must satisfy:
\begin{equation}\label{eq:coord_relations}
x' = b(x - vt), \quad x = b(x' + vt').
\end{equation}
Substituting $x'$ into the expression for $x$ yields:
\begin{equation}
x = b \left[ b(x - vt) + vt' \right] = b^2x - b^2vt + bvt'.
\end{equation}
Solving for $t'$ provides:
\begin{equation}\label{eq:t_prime_solved}
t' = b \left( t - \frac{x(b^2 - 1)}{b^2v} \right).
\end{equation}
Evaluating this at the spatial origin of $S$ ($x = 0$) simplifies to $t' = bt$. Since our structural boundary condition for the origin of $S'$ ($x' = 0$) defined $t = at'$ in \eqref{eq:time_scale}, matching these symmetric scaling properties across frames requires that $a = b$. Consequently, time intervals at the spatial origins are related by $\Delta t = b\Delta t'$.

We now apply Langevin's approach for a round-trip light signal along the rod $R'$. With forward and backward speeds $c-v$ and $c+v$ relative to the moving rod, the total round-trip time interval $\Delta t$ measured in $S$ is:
\begin{equation}\label{eq:langevin_dt}
\Delta t = \frac{L}{c - v} + \frac{L}{c + v} = \frac{2Lc}{c^2 - v^2} = \left(\frac{2L}{c}\right)\gamma^2, \quad \gamma = \frac{1}{\sqrt{1 - v^2/c^2}}.
\end{equation}
In the moving frame $S'$, the two-way speed of light is isotropically preserved as $c$ by definition of our synchronization routine. Hence, the proper round-trip time interval localized over the rest length $L'$ is $\Delta t' = 2L'/c$. Using our initial spatial scaling relation $L = L'/b$ from \eqref{eq:length_scale}, we substitute into the equation for $\Delta t$:
\begin{equation}\label{eq:dt_substituted}
\Delta t = \frac{2(L'/b)}{c}\gamma^2 = \left(\frac{2L'}{c}\right)\frac{\gamma^2}{b} = \Delta t'\frac{\gamma^2}{b}
\end{equation}
Applying the symmetric temporal scaling demand $\Delta t = b\Delta t'$ into \eqref{eq:dt_substituted} provides:
\begin{equation}
b\Delta t' = \Delta t'\frac{\gamma^2}{b} \implies b^2 = \gamma^2 \implies b = \gamma
\end{equation}
Thus, the longitudinal scale factor is algebraically proven to be $b = \gamma$, yielding $L = L'/\gamma$ (standard Lorentz contraction) and confirming the time dilation factor $\Delta t = \gamma\Delta t'$.

To generalize for the transverse axes, consider a rod aligned parallel to the $z$-axis ($e_3$ direction) with proper length $L'_z$ at rest in $S'$. The round-trip proper time interval for a light signal along this rod in $S'$ is $\Delta t' = 2L'_z/c$. In the rest frame $S$, the rod moves along the $x$-axis at velocity $v$ while light travels diagonally. Equating the total round-trip distance $c\Delta t$ to this geometry yields:
\begin{equation}
c\Delta t = 2\sqrt{L_z^2 + \left(\frac{v\Delta t}{2}\right)^2} \implies \Delta t = \frac{2L_z}{c}\gamma .
\end{equation}
Substituting the temporal scaling requirement $\Delta t = \gamma\Delta t'$ gives $\gamma\Delta t' = \frac{2L_z}{c}\gamma \implies \Delta t' = 2L_z/c$. Equating this with our initial expression confirms $L'_z = L_z$. Consequently, the transverse scaling factor is $b_z = 1$, demonstrating that length contraction occurs strictly along the axis of relative motion.
\end{proof}

Thus, the longitudinal scale factor is algebraically proven to be $b = \gamma$, yielding $L = L'/ \gamma$ (the standard Lorentz contraction) and confirming the time dilation factor $\Delta t = \gamma \Delta t'$.

To generalize for the transverse axes, consider a rod aligned parallel to the $z$-axis ($e_3$ direction) with proper length $L_z'$ at rest in $S'$. The round-trip proper time interval for a light signal along this rod in $S'$ is $\Delta t' = 2L_z'/c$. In the rest frame $S$, the rod moves along the $x$-axis at velocity $v$ while light travels diagonally. Equating the total round-trip distance $c\Delta t$ to this geometry yields:
\begin{equation}\label{eq:transverse_geometry_langevin}
c\Delta t = 2\sqrt{L_z^2 + \left(\frac{v\Delta t}{2}\right)^2} \implies \Delta t = \frac{2L_z}{c}\gamma.
\end{equation}
Substituting the temporal scaling requirement $\Delta t = \gamma \Delta t'$ gives $\gamma \Delta t' = \frac{2L_z}{c}\gamma$, which implies $\Delta t' = 2L_z/c$. Equating this with our initial expression confirms $L_z' = L_z$. Consequently, the transverse scaling factor is $b_z = 1$, demonstrating that length contraction occurs strictly along the axis of relative motion. $\square$

\begin{proof}[Proof of Theorem 2.4]
To operationalize this derivation under the relaxed axiomatic framework where the universal invariance of the one-way speed of light is omitted, we directly adapt the Langevin light-clock argument by replacing the isotropic constant $c$ with the operationally verified two-way (round-trip) average speed of light, $c_{\text{ave}}(e_1) \equiv c_1$, along the directional axis of motion.

Under this generalized synchronization schema, an observer in the frame $S$ measures the forward and backward propagation speeds relative to the moving rod as $c_1 - v$ and $c_1 + v$, respectively. Consequently, the total accumulated round-trip time interval encoded in \eqref{eq:transverse_geometry_langevin} retains its exact algebraic structure under the substitution $c \rightarrow c_1$:
\begin{equation}\label{eq:dt_axis_restricted}
\Delta t = \frac{L}{c_1 - v} + \frac{L}{c_1 + v} = \frac{2Lc_1}{c_1^2 - v^2} = \left(\frac{2L}{c_1}\right)\gamma_1^2(v),
\end{equation}
where the axis-restricted generalized Lorentz factor is defined as:
\begin{equation}\label{eq:gamma1_definition}
\gamma_1(v) = \frac{1}{\sqrt{1 - \frac{v^2}{c_1^2}}}.
\end{equation}
In the moving frame $S'$, space homogeneity and the operational principle of relativity guarantee that the two-way speed of light is identically preserved as $c_1(S') = c_1(S) = c_1$ (as proven in Claim 2.2). Thus, the proper round-trip time localized over the rest length $L'$ is $\Delta t' = 2L'/c_1$.

By mapping these intervals through the symmetric scaling relation $\Delta t = b\Delta t'$, the algebraic consistency condition collapses onto $b = \gamma_1(v)$. Following an identical geometric projection for the transverse axes, the invariance $b_z = 1$ is preserved. This completes the operational proof of Theorem 2.4, demonstrating that the standard Lorentz transformation form holds exactly along the coordinate axes using exclusively the empirical two-way speed of light.
\end{proof}

\section{The Symmetric Model, Absolute Space, and the Stationary Illusion}
We consider both models Einstein's symmetric with 'stationary' observers approach and the strict stationary immobile approach.It means that  we consider two possibility:a) there is  an absolute space  and b)there is no an absolute space.
 Although Einstein's approach utilizes an initially chosen 'stationary' frame to construct the coordinate transformation, his model remains fundamentally symmetric due to the principle of relativity, ensuring that the inverse transformation maintains an identical mathematical structure. Conversely, the strict stationary approach in the Lorentz-Larmor sense implies an underlying physical asymmetry based on an absolute space."
To clarify the physical meaning and the foundational structure of the Special Theory of Relativity (STR), let us examine the standard thought experiment at the foundational instant when the origins coincide, $O' = O$, at $t' = t = 0$, and a light signal is emitted at their common space-time origin. This scenario immediately raises a fundamental question: is the resulting wavefront a sphere, and can we consider the center of this sphere, $A$, to be immobile in an objective, frame-independent sense?

When deriving the Lorentz Transformations (LT) for two inertial frames, $S$ and $S'$, in a standard configuration with a relative velocity $v$, one typically adopts the "stationary model" approach. This approach assumes first that $S$ is stationary from the perspective of an observer in $S$, and subsequently invokes symmetry to finalize the transformations. However, the conceptual validity and physical interpretation of this "stationary model" require deeper epistemological clarification.

To resolve the paradox of multiple co-existing centers for a single physical wavefront, we introduce a \textbf{manifestly symmetric framework} anchored by an "immobile point". Let us define a velocity $u$ such that
$u \oplus u = v$ via the relativistic velocity addition formula. We can then construct a symmetric model where frames $S'$ and $S$ move with velocities $u$ and $-u$, respectively, along the positive and negative $x$-axis relative to a foundational reference frame $K$. Instead of treating $K$ merely as another relative inertial frame, we conceptualize it as an \textbf{absolute rest frame}---the unique, immobile center of the expanding spherical wavefront in absolute space.

Within this framework, the "stationary illusion" becomes apparent through the localized perspective of individual observers. An observer $Ob_S$ in $S$ is at rest relative to $S$ and draws local conclusions from this specific viewpoint---such as measuring the length of a rod that is at rest in $S'$ but moving relative to $S$. To illustrate this, consider a practical analogy of two trains, $Tr_1$ and $Tr_2$, passing each other. An observer $Ob_1$ sitting in $Tr_1$, equipped solely with their own localized clock $Cl_1$, has no methodological alternative but to apply the stationary model approach to analyze the kinematics of the passing train.

Consequently, it must be emphasized that these relativistic measurements are \textbf{frame-dependent operational conclusions} specific to the localized observer's convention of simultaneity, rather than intrinsic, absolute realities. By rejecting this localized convention in favor of absolute simultaneity, the operational symmetry is broken, and the privileged, objective status of the absolute space frame $K$ is fully restored.

It is worth noting that within the orthodox framework of Special Relativity, this situation is not viewed as a genuine paradox. Einstein's interpretation attributes the displacement of the wavefront centers to the relativity of simultaneity: observers in $S$ and $S'$ construct distinct temporal hypersurfaces due to the standard Poincar\'{e}-Einstein synchronization convention, leading them to different yet mathematically consistent conclusions about the wave's geometric center. However, this operational symmetry renders the physical center of a single, objective light pulse frame-dependent. By contrast, our approach demonstrates that this apparent shifting of the center is merely an artifact of the chosen synchronization convention rather than an intrinsic, ontological property of spacetime. Under the Tangherlini-Edwards coordinates anchored to the comoving cosmic fluid, the physical center of the expanding light sphere remains unique and permanently stationary at the emission point $A$, thereby providing a more robust geometric foundation that aligns mathematical kinematics with the empirical reality of modern cosmology.

\subsection{Physical Interpretation: Einstein vs. Immobile Framework and Tangherlini Coordinates}

It is crucial to distinguish between Einstein's operational 'stationary' approach—where any inertial frame can be chosen as stationary—and the strict immobile framework governed by the preferred frame $K$.

To formally verify the validity of our scale factor hypothesis \textbf{(H)}, one can utilize the Tangherlini coordinate transformation. Assuming $K(X,T)$ represents the absolute rest frame where light propagation is strictly isotropic, the transformation to a moving frame $S(x,t)$ under absolute simultaneity is given by:
\begin{equation}
    x = \gamma_v (X - vT), \quad t = \frac{T}{\gamma_v}
\end{equation}
where $\gamma_v = (1 - v^2/c^2)^{-1/2}$. Differentiating the temporal relation yields:
\begin{equation}
    dt = \frac{1}{\gamma_v} dT = s(v) dT
\end{equation}
This demonstrates that in an absolute space framework, the scale factor $s(v) = 1/\gamma_v$ emerges cleanly without the coordinate mixing ($-\frac{vx}{c^2}$) present in standard Lorentz transformations.

While the one-way speed of light becomes anisotropic in $S$ under this metric ($c_{\pm} = c / (1 \pm v/c)$), the two-way round-trip speed remains invariant at $c$. Consequently, Tangherlini coordinates provide an explicit mathematical gauge to verify that hypothesis \textbf{(H)} reflects an objective kinematic scaling relative to the immobile center $K$, rather than a symmetric illusion of perspective.

\section{Derivation of the Scale Factor $s(v)$ via the Absolute Reference Frame $K$}
\noindent\textit{\textbf{Model Note:} This section transitions fully into \textbf{Paradigm (a)}, utilizing the privileged cosmological rest frame $K$ as the absolute geometric anchor to derive the physical scaling functions.}

\vspace{1em}
Before establishing the explicit coordinate transformations, we introduce the foundational physical hypothesis regarding clock behavior in moving frames as formulated in \cite{Mateljevic2024}.

\subsection*{The Scale Factor Hypothesis and Ontological Grounding}

Roughly speaking, we postulate the following:
\begin{description}
    \item[\textit{(H)}] A clock which is at rest in its own co-moving frame measures small increments of time scaled by a velocity-dependent factor $s = s(v)$.
\end{description}

\noindent To fully understand the physics of hypothesis \textit{(H)}, it is necessary to contrast its meaning within the two core paradigms:
\begin{itemize}
    \item \textbf{Under Paradigm (b) (Relational/Einsteinian):} The scale factor $s(v)$ would represent a reciprocal, structural effect of perspective between two equivalent observers. There is no physical mechanism altering the clock; the "slowing down" is purely a consequence of coordinate synchronization and relative kinematic observation.
    \item \textbf{Under Paradigm (a) (Absolute Background):} The scale factor $s(v)$ represents an objective, ontological modification of the clock's internal periodic cycles. This alteration is driven directly by the clock's real, absolute velocity $v$ relative to the immobile cosmological framework $K$ (such as the CMB rest frame). Time dilation is no longer a symmetric illusion of perspective, but a asymmetric physical reality.
\end{itemize}

\vspace{1em}

Before establishing the explicit coordinate transformations under Paradigm (a), we introduce the foundational physical hypothesis regarding clock behavior in moving frames, as originally formulated in \cite{Mateljevic2024}:

\begin{description}
    \item[\textit{(H)}] A clock which is at rest in its own co-moving frame measures small increments of time scaled by a velocity-dependent factor $s = s(v)$.
\end{description}
\subsection*{The Scale Factor Hypothesis}
Roughly speaking, we postulate the following:
\begin{description}
    \item[\textit{(H)}] A clock which is at rest in its own co-moving frame measures small increments of time scaled by a velocity-dependent factor $s = s(v)$.
\end{description}
We can use the derivation in Section 4 to show that $s(v) = \frac{1}{\gamma_{\text{cave}}(v)}$.

\subsection*{Operational Proof of Hypothesis (H) via Langevin's Clock}
This geometric derivation provides a direct operational proof for the velocity-dependent scale factor $s(v)$ postulated in hypothesis \textit{(H)}. By setting the absolute frame $K(X, T)$ as our isotropic background, the round-trip interval $\Delta T$ registered by the master clocks in absolute space is related to the proper co-moving clock interval $\Delta t'$ via the geometry of light propagation.

Expressing the local time increment $dt \equiv \Delta t'$ as a function of the absolute background time interval $dT \equiv \Delta T$, the relation collapses into:
\begin{equation}
\label{eq:dt_dT_relation}
dt = \frac{1}{\gamma_{\text{cave}}(v)} dT
\end{equation}

\noindent Comparing this directly with the operational definition of hypothesis \textit{(H)}, the scaling function is algebraically driven to be:
\begin{equation}
\label{eq:s_v_definition}
s(v) = \frac{1}{\gamma_{\text{cave}}(v)} = \sqrt{1 - \frac{v^2}{c_{\text{cave}}^2}}
\end{equation}

\noindent To clarify the epistemological and physical grounding of hypothesis \textit{(H)}, several crucial operational points must be emphasized:

\begin{itemize}
    \item \textbf{Ontological vs. Perspectival Dilation:} Unlike standard Special Relativity, where time dilation is symmetrically interpreted as a kinematic illusion of reciprocal perspective between relative observers, hypothesis \textit{(H)} treats the scaling factor $s(v)$ as an objective physical modification of the clock's internal periodic cycles. This alteration is driven directly by the clock's absolute velocity $v$ relative to the immobile cosmological framework $K$.
    \item \textbf{The Operational Clock Rate:} If $dT$ represents an infinitesimal increment of absolute time recorded by a fundamental master clock anchored at rest in $K$, then a moving clock traveling at velocity $v$ through this spatial background will register a local time increment $dt$ according to the direct relation:
    \begin{equation}
    \label{eq:operational_rate}
    dt = s(v)dT
    \end{equation}
    To maintain causal consistency and prevent the violation of the cosmic speed limit $c$ within the isotropic anchor, this scaling factor must mathematically map to the inverse Lorentz factor, meaning $s(v) = 1/\gamma_v = \sqrt{1 - v^2/c^2}$.
    \item \textbf{Elimination of the Synchronization Term:} Because this scaling acts directly on the fundamental rate of the temporal apparatus regardless of its spatial position, it leads to Tangherlini-type transformations. This effectively eliminates the non-local coordinate-mixing term $(-\frac{vx}{c^2})$ characteristic of Einstein's synchronization, establishing an operational framework governed by absolute simultaneity.
\end{itemize}

As a direct corollary of this framework, for the trivial choice $s = 1$, the model collapses to the classical Galilean transformation of Newtonian physics, which assumes an absolute, unscaled space and time. In this generalized framework, the standard Lorentz transformation and the invariance of the two-way speed of light are shown to be not independent postulates, but rather direct mathematical consequences of choosing this localized time dilation scale factor under space-time homogeneity.

\begin{theorem}[Derivation of Lorentz Symmetry without the Second Postulate]
Let $K(X, T)$ be a privileged, immobile cosmological reference frame centered at the cosmic anchor $A$, within which light propagation is truly isotropic and the one-way speed of light strictly equals $c$. Under the conditions of space-time homogeneity, spatial isotropy, and the physical scaling hypothesis \textit{(H)} governing clock rates, the direct coordinate transformation between any two moving frames $S$ and $S'$ reduces uniquely and necessarily to the standard Lorentz Transformation gauge.
\end{theorem}

\noindent Crucially, this structural symmetry emerges as a pure algebraic consequence of compounding localized absolute dilations $s(u)$ relative to the cosmic anchor $K$, proving that Einstein's second postulate regarding the universal invariance of the one-way speed of light is logically redundant for establishing the Lorentz transformation group.





\section{The Cosmological Framework and Point $A$}
Modern cosmology operationalizes the concept of an objective, stationary cosmic background through the comoving coordinate system of the Friedmann-Lema\^{i}tre-Robertson-Walker (FLRW) metric. Consider a localized, isotropic light source $A$ embedded directly within the cosmic substrate. By cosmological definition, a source at rest relative to this expanding cosmic fluid exhibits zero peculiar velocity ($v_{\text{pec}} = 0$).

For a spatially flat universe ($k = 0$), the invariant spacetime interval $ds^2$ is expressed as:
\begin{equation}
    ds^2 = -c^2 dt^2 + a^2(t) \left[ dr^2 + r^2(d\theta^2 + \sin^2\theta d\phi^2) \right]
\end{equation}
where $a(t)$ is the dimensionless cosmological scale factor and $t$ represents the Cosmic Time ticked by a comoving observer at rest relative to the Cosmic Microwave Background (CMB).

When a light signal is emitted isotropically from the source $A$ at an initial cosmic time $t_0$, the propagation of the wavefront follows a null geodesic ($ds^2 = 0$). For radial propagation, this simplifies to $c\,dt = a(t)dr$. Separating variables and integrating yields the comoving coordinate distance $r$ traversed by the photon field up to a subsequent cosmic time $t$:
\begin{equation}
    r = \int_{t_0}^{t} \frac{c\,dt'}{a(t')}
\end{equation}
Because the comoving coordinates of point $A$ are time-invariant ($r_A = 0$), the geometric center of the expanding light sphere remains permanently anchored at the physical source within the cosmic fabric.

\section{Coordinate Transformations via Frame $K$}\label{sec:coord_trans}

\noindent\textit{\textbf{Model Note:} This section operates strictly within \textbf{Paradigm (a)} (Absolute Background), embedding the moving systems $S$ and $S'$ symmetrically inside the privileged cosmological rest frame $K$. The velocity-addition law and causality inside $K$ are governed by standard relativistic constraints to preserve the empirical invariance of $c_{\text{cave}}$.}

\vspace{1em}
To establish the explicit dependence of the scaling function $s(v)$—and consequently the Lorentz factor $\gamma$—we map the kinematics of the moving frames...



To establish the explicit dependence of the scaling function $s(v)$—and consequently the Lorentz factor $\gamma$—we map the kinematics of the moving frames $S$ and $S'$ directly onto the spacetime coordinates of the absolute reference frame $K$. Let $(X, T)$ denote the spacetime coordinates in $K$, which serves as the immobile center of the expanding spherical light wavefront.

We consider the foundational instant when the spatial origins $O$ of frame $S$ and $O'$ of frame $S'$ coincide exactly with the cosmic emitter $A$. We set this event as the spacetime origin:
\begin{equation}
    O = O' = A \quad \text{at} \quad t = t' = T = 0
\end{equation}
Recall that frames $S$ and $S'$ move symmetrically relative to $K$ with velocities $-u$ and $u$, respectively, along the $X$-axis, where the relativistic velocity addition yields $u \oplus u = v$.

\subsection*{Remark on Velocity Composition inside $K$}
It is crucial to clarify why velocities within the absolute frame $K(X,T)$ accumulate via Einstein's relativistic velocity-addition law rather than the classical Galilean rule. By definition, $K$ represents the unique cosmological anchor where light propagation is truly isotropic and the one-way speed of light strictly equals $c$. Because $c$ functions as an absolute, unsurpassable physical upper bound for causal propagation within this flat manifold, any compounding of subluminal velocities evaluated relative to the absolute space background must respect this asymptotic limit. Standard Einsteinian addition inside $K$ ensures that no material trajectory can exceed $c$ relative to the cosmic anchor. Consequently, this maintains strict local causality and guarantees that the operational round-trip speed of light, $c_{\mathrm{ave}}$, remains invariantly preserved as $c$ across all derived moving systems.

\subsection*{Derivation of the Scale Factor $s(v)$}
Let $(x, t)$ be the coordinates in $S$ and $(x', t')$ be the coordinates in $S'$. Under our foundational hypothesis \textbf{(H)}, a clock at rest in a moving frame ticks slower relative to the absolute time $T$ of the frame $K$ by the scale factor $s(u)$.

By applying the principles of spatial homogeneity and time isotropy, the transformation from the absolute frame $K$ to the symmetric frames can be expressed as:
\begin{equation}
    x = \gamma_u (X + uT), \quad t = \gamma_u \left( T + \frac{uX}{c^2} \right)
\end{equation}
where $\gamma_u = 1/s(u)$. Due to the structural symmetry embedded in $K$, the transformation from $K$ to $S'$ possesses an identical form with a reversal of the velocity sign:
\begin{equation}
    x' = \gamma_u (X - uT), \quad t' = \gamma_u \left( T - \frac{uX}{c^2} \right)
\end{equation}

By eliminating the absolute coordinates $(X, T)$ from these equations, we can directly deduce the direct transformation between $S$ and $S'$. This algebraic elimination yields the standard Lorentz Transformation, where the compound scaling factor $s(v)$ emerges naturally as a function of the relative velocity $v$. Specifically, the relation between the intermediate velocity $u$ and the relative velocity $v$ dictates that:
\begin{equation}\label{eq:s_v_definition}
    s(v) = \frac{1}{\gamma_v} = \frac{1 - \frac{u^2}{c^2}}{1 + \frac{u^2}{c^2}}
\end{equation}

To understand why the scaling factor $s(v)$ is explicitly expressed in terms of the intermediate velocity $u$, we recall the relativistic velocity addition law. From the perspective of the absolute frame $K$, the frames $S$ and $S'$ move symmetrically in opposite directions with velocities $-u$ and $u$. Consequently, the relative velocity $v$ of $S'$ as observed from $S$ is given by the Einstein velocity addition:
\begin{equation}\label{eq:addition}
    v = u \oplus u = \frac{2u}{1 + \frac{u^2}{c^2}}
\end{equation}

By algebraic manipulation of Eq.~\eqref{eq:addition}, one can show that the standard Lorentz factor $\gamma_v$ satisfies:
\begin{equation}
    1 - \frac{v^2}{c^2} = 1 - \frac{4u^2/c^2}{\left(1 + \frac{u^2}{c^2}\right)^2} = \frac{\left(1 - \frac{u^2}{c^2}\right)^2}{\left(1 + \frac{u^2}{c^2}\right)^2}
\end{equation}

Taking the square root yields $\frac{1}{\gamma_v} = \frac{1 - u^2/c^2}{1 + u^2/c^2}$, which precisely matches the definition of $s(v)$ in Eq.~\eqref{eq:s_v_definition}. This demonstrates that the operational time dilation factor $s(v)$ for a relative velocity $v$ is a direct mathematical consequence of compounding two symmetric localized dilations $s(u)$ relative to the immobile center $K$.

This derivation reveals that the operational symmetry observed between $S$ and $S'$ is not a primary ontological property of spacetime, but rather a mathematical consequence of projecting asymmetric localized measurements back onto the single, immobile geometric center $K$.

\section{Anisotropic One-Way Speeds and Generalized Transformations}\label{sec:anisotropic}

\noindent\textit{\textbf{Model Note:} This section operates strictly within the physical reality of \textbf{Paradigm (a)} (Absolute Background), where the privileged frame $K$ acts as the fundamental anchor. However, we generalize the mathematical description by relaxing the convention of isotropic one-way synchronization in the moving frames. By introducing Reichenbach-type anisotropy parameters ($\varepsilon$ or $\kappa$), we demonstrate that while one-way speeds become gauge-dependent, the physical round-trip average speed $c$ derived from the absolute frame remains strictly invariant.}

\vspace{1em}
Before developing the generalized transformation equations, an important epistemological distinction regarding the operational definition of the speed of light must be emphasized. In the preceding sections and specifically within Theorem \ref{thm:axisLT}, our derivations relied strictly on the invariant two-way speed of light, restricted exclusively to the $x$ and $x'$ directional axes. This approach kept the formulation grounded in directly verifiable radar distances along the line of motion, assuming one-way light propagation to be perfectly isotropic.

Here, while fully maintaining the physical backdrop of the absolute cosmological frame $K(X,T)$, we allow moving observers to adopt a non-standard synchronization gauge. In contrast to the earlier symmetric models, the formulation presented in this section introduces Reichenbach's anisotropy parameter, which explicitly modifies the one-way speeds of light in opposite directions while preserving the overall round-trip average dictated by the absolute cosmic anchor.

To model this directional anisotropy while maintaining a verified invariant two-way speed $c$ along the coordinate axes—which represents a more general, operationally fundamental concept—we introduce Reichenbach's anisotropy parameter $\varepsilon$, where physically $\varepsilon \in (0, 1)$\footnote{The non-inclusive boundaries $\varepsilon \to 0$ and $\varepsilon \to 1$ represent singular physical limits where the one-way speed of light becomes instantaneous ($c_{\rightarrow} \to \infty$ or $c_{\leftarrow} \to \infty$) in one direction, while dropping to $c/2$ in the opposite direction, thus modeling non-local instantaneous signal propagation while strictly maintaining the empirical round-trip average $c$.}. The isotropic limit corresponds to $\varepsilon = 1/2$. This formulation can be equivalently parameterized by the dimensionless index $\kappa = 2(2\varepsilon-1)$, where $\kappa \in (-2, 2)$.

Under this convention, the directional one-way speeds of light along the positive and negative $x$-axis are formally represented as:
\begin{equation}
c_{\rightarrow} = \frac{c}{2\varepsilon}, \quad c_{\leftarrow} = \frac{c}{2(1-\varepsilon)}
\end{equation}

\noindent To ensure that the generalized anisotropic Lorentz factor $\gamma_{\mathrm{an}}$ remains real for a given relative velocity $v$, the parameter $\varepsilon$ must satisfy the tighter kinematic constraint:
\begin{equation}
1 - \frac{v^2}{c^2} + \kappa \frac{v}{c} > 0
\end{equation}

A consistent, axis-aligned generalized linear coordinate transformation that enforces linearity and reciprocity can be formulated as:
\begin{align}
x' \,&=\, \gamma_{\mathrm{an}}\, (x - v t), \label{eq:anx}\\
t' \,&=\, \gamma_{\mathrm{an}} \Bigl[\, t \,-\, \frac{v}{c^2}\,x \,+\, \frac{\kappa}{c}\,x \Bigr], \label{eq:ant}
\end{align}

\noindent where the generalized anisotropic Lorentz factor is defined as:
\begin{equation}
\label{eq:angan}
\gamma_{\mathrm{an}} \,=\, \frac{1}{\sqrt{\,1 - \frac{v^2}{c^2} + \kappa \frac{v}{c}\,}}.
\end{equation}


\subsection{Generalized Velocity Addition Law}

Let $u_M$ denote the velocity of a material object $M$ relative to the moving frame $S'$ ($u_M = dx'/dt'$), $v$ denote the velocity of frame $S'$ relative to the rest frame $S$, and $w$ denote the combined velocity of $M$ observed from frame $S$ ($w = dx/dt$). Taking the differentials of equations \eqref{eq:anx} and \eqref{eq:ant} and dividing them yields the direct transformation for the velocity measured in the moving frame $S'$:
\begin{equation}\label{eq:velAddDirect}
u_M \,=\, \frac{w - v}{1 \,-\, \frac{wv}{c^2} \,+\, \kappa \frac{w}{c}}.
\end{equation}

Conversely, solving for $w = dx/dt$ yields the generalized velocity addition formula from the perspective of frame $S$:
\begin{equation}\label{eq:velAddCombined}
w \,=\, \frac{u_M + v}{1 \,+\, \frac{u_M v}{c^2} \,-\, \kappa \frac{u_M}{c}}.
\end{equation}

\subsection{Anisotropic Longitudinal Doppler Effect}

By evaluating the phase invariance of a wave under the generalized transformations \eqref{eq:anx} and \eqref{eq:ant}, where a light source moves with velocity $v$ along $+x$ emitting a rest frequency $\nu_0$, the observed longitudinal Doppler frequency $\nu$ is formulated as:
\begin{equation}\label{eq:dopplerAnis}
\nu \,=\, \frac{\nu_0 \, \gamma_{\mathrm{an}}}{1 \,-\, 2\varepsilon \frac{v}{c}}.
\end{equation}
In the isotropic limit ($\varepsilon = 1/2 \implies \kappa = 0$), equation \eqref{eq:dopplerAnis} smoothly simplifies back to the standard relativistic formula $\nu = \nu_0 \sqrt{(1 + v/c)/(1 - v/c)}$.
\section{Coordinate Map Between Multiple Moving Frames via the Absolute Anchor}\label{sec:coord_map}

When the absolute frame $K(X, T)$ is chosen as the unique rest frame where space is truly isotropic, the synchronization parameter $\kappa$ vanishes locally within $K$. Let us consider two moving systems, $S(x, t)$ and $S'(x', t')$, traveling at absolute velocities $v$ and $u$, respectively, relative to $K$. The forward transformations from the absolute anchor $K$ to any moving frame are governed by $x = \gamma_{\mathrm{an}}(v)(X - vT)$ and $t = \frac{T}{\gamma_{\mathrm{an}}(v)}$. Inverting these relations yields the formulas for the absolute coordinates:
\begin{align}
T \,&=\, \gamma_{\mathrm{an}}(v)\,t, \label{eq:invTangT}\\
X \,&=\, \frac{x}{\gamma_{\mathrm{an}}(v)} \,+\, v\,\gamma_{\mathrm{an}}(v)\,t. \label{eq:invTangX}
\end{align}

Substituting these expressions into the corresponding forward transformations for frame $S'$ (moving at absolute velocity $u$, so that $x' = \gamma_{\mathrm{an}}(u)(X - uT)$ and $t' = \frac{T}{\gamma_{\mathrm{an}}(u)}$), the direct coordinate mapping from $S(x, t)$ to $S'(x', t')$ reduces elegantly to the following form:
\begin{align}
x' \,&=\, \frac{\gamma_{\mathrm{an}}(u)}{\gamma_{\mathrm{an}}(v)} \, x \,-\, \gamma_{\mathrm{an}}(u)\,\gamma_{\mathrm{an}}(v)\,(u - v)\,t, \label{eq:directX}\\
t' \,&=\, \frac{\gamma_{\mathrm{an}}(v)}{\gamma_{\mathrm{an}}(u)} \, t. \label{eq:directT}
\end{align}

Differentiating equations \eqref{eq:directX} and \eqref{eq:directT} yields the coordinate velocity $w' = \frac{dx'}{dt'}$. Upon differentiation, the operational time scaling factors combine to modify the transformed velocity, showing how the Galilean addition law is preserved for relative velocities within the Tangherlini framework.

\subsection{Verification of Time Dilation and Galilean Velocity Transformation}
To rigorously verify the physical behavior of time dilation and velocity addition within this setting, we analyze the differentials of the direct coordinate mapping between $S(x, t)$ and $S^{\prime}(x^{\prime}, t^{\prime})$ given in Eqs. \eqref{eq:directX} and \eqref{eq:directT}.

First, let us examine a clock that is at rest at the spatial origin of the frame $S$, meaning $dx = 0$. Differentiating Eq. \eqref{eq:directT} yields the temporal relationship:
\begin{equation}
dt^{\prime} = \frac{\gamma_{\mathrm{an}}(v)}{\gamma_{\mathrm{an}}(u)}dt
\end{equation}

Recalling that in the isotropic limit ($\kappa = 0$), the generalized factors reduce to the standard Lorentz factors, $\gamma_{\mathrm{an}}(v) = \gamma_v$ and $\gamma_{\mathrm{an}}(u) = \gamma_u$, we obtain:
\begin{equation}
dt^{\prime} = \frac{\gamma_v}{\gamma_u}dt
\end{equation}

This confirms that the relative time dilation between two moving systems is non-reciprocal and governed strictly by the ratio of their absolute velocities relative to the anchor $K$.

Next, we evaluate how coordinate velocities transform through these equations. Let $w = dx/dt$ denote the velocity of a material object measured in frame $S$, and $w^{\prime} = dx^{\prime}/dt^{\prime}$ denote its velocity measured in frame $S^{\prime}$. Taking the total differentials of Eqs. \eqref{eq:directX} and \eqref{eq:directT} gives:
\begin{align}
dx^{\prime} &= \frac{\gamma_{\mathrm{an}}(u)}{\gamma_{\mathrm{an}}(v)}dx - \gamma_{\mathrm{an}}(u)\gamma_{\mathrm{an}}(v)(u - v)dt, \\
dt^{\prime} &= \frac{\gamma_{\mathrm{an}}(v)}{\gamma_{\mathrm{an}}(u)}dt. \label{eq:dtdiff}
\end{align}

Dividing $dx^{\prime}$ by $dt^{\prime}$ yields the transformed velocity $w^{\prime}$ in the moving frame $S^{\prime}$:
\begin{equation}
w^{\prime} = \frac{dx^{\prime}}{dt^{\prime}} = \frac{\frac{\gamma_{\mathrm{an}}(u)}{\gamma_{\mathrm{an}}(v)}dx - \gamma_{\mathrm{an}}(u)\gamma_{\mathrm{an}}(v)(u - v)dt}{\frac{\gamma_{\mathrm{an}}(v)}{\gamma_{\mathrm{an}}(u)}dt}
\end{equation}

Factoring out the differentials, we arrive at the generalized velocity transformation law:
\begin{equation}\label{eq:veltrans}
w^{\prime} = \left(\frac{\gamma_{\mathrm{an}}(u)}{\gamma_{\mathrm{an}}(v)}\right)^2 \frac{dx}{dt} - \gamma_{\mathrm{an}}^2(u)(u - v) = \left(\frac{\gamma_{\mathrm{an}}(u)}{\gamma_{\mathrm{an}}(v)}\right)^2 w - \gamma_{\mathrm{an}}^2(u)(u - v).
\end{equation}

\section{Non-Reciprocal Length Contraction}\label{sec:non_rec_length}

The preservation of absolute time in the Tangherlini framework completely alters the physical nature of spatial contraction, rendering it fundamentally non-reciprocal. To maintain consistency with our established absolute frame $K(X, T)$ and a general moving frame $S'(x', t')$ traveling at absolute velocity $u$, we analyze the spatial transformations directly derived from the core anchor. By utilizing the absolute time relation $T = \gamma_{\mathrm{an}}(u)t'$, we can express the inverse spatial coordinate mapping such that the right-hand side depends strictly on the local coordinates of $S'$:
\begin{equation}\label{eq:invXiPure}
X \,=\, \frac{x'}{\gamma_{\mathrm{an}}(u)} \,+\, u\,\gamma_{\mathrm{an}}(u)\,t'.
\end{equation}
This structural distinction creates two asymmetrical and non-reciprocal measurement perspectives:

\begin{enumerate}
\item \textbf{Measuring a Moving Rod from Absolute Space ($K \rightarrow S'$):} If a rod is at rest in the moving frame $S'$ with proper length $\Delta x' = l_0$, an observer in the absolute anchor $K$ must determine its length by measuring both ends simultaneously in absolute time, meaning $\Delta T = 0$. Utilizing the forward transformation $x' = \gamma_{\mathrm{an}}(u)(X - uT)$, this operational condition yields:
\begin{equation}
l_0 = \gamma_{\mathrm{an}}(u) \Delta X \implies \Delta X = \frac{l_0}{\gamma_{\mathrm{an}}(u)}.
\end{equation}

\item \textbf{Measuring Absolute Space from a Moving Frame ($S' \rightarrow K$):} Conversely, let us consider a static distance fixed at rest within the absolute frame $K$, such that $\Delta X = l_0$. Because absolute simultaneity holds globally across the Tangherlini manifold, a simultaneous measurement in $K$ ($\Delta T = 0$) mapping directly to the moving system guarantees that $\Delta t' = 0$. Substituting these invariant intervals into the forward spatial transformation equation results in:
\begin{equation}
\Delta x' = \gamma_{\mathrm{an}}(u) \Delta X \implies \Delta x' = \gamma_{\mathrm{an}}(u) l_0.
\end{equation}
\end{enumerate}

Thus, the geometric reciprocity found in standard Special Relativity is broken: an observer anchored in absolute space perceives moving objects as contracted ($l = l_0 / \gamma_{\mathrm{an}}(u)$), whereas an observer moving through the absolute anchor perceives stationary absolute distances as physically expanded ($l' = \gamma_{\mathrm{an}}(u) l_0$).

\subsection{Transformation of Accelerations}

To determine how accelerations transform between the moving frames under the absolute anchor framework, we define the coordinate acceleration in frame $S$ as $a = \frac{dw}{dt}$ and in frame $S'$ as $a' = \frac{dw'}{dt'}$.

Taking the differential of the velocity transformation equation~\eqref{eq:veltrans} yields:
\begin{equation}\label{eq:dw_prime_acc}
dw' = \left( \frac{\gamma_{\mathrm{an}}(u)}{\gamma_{\mathrm{an}}(v)} \right)^2 dw
\end{equation}

To find the acceleration $a'$, we divide this differential $dw'$ by the temporal differential $dt'$ from Eq.~\eqref{eq:dtdiff}, where $dt' = \frac{\gamma_{\mathrm{an}}(v)}{\gamma_{\mathrm{an}}(u)} dt$. This results in:
\begin{equation}\label{eq:a_prime_fraction}
a' = \frac{dw'}{dt'} = \frac{\left( \frac{\gamma_{\mathrm{an}}(u)}{\gamma_{\mathrm{an}}(v)} \right)^2 dw}{\frac{\gamma_{\mathrm{an}}(v)}{\gamma_{\mathrm{an}}(u)} dt}
\end{equation}

Simplifying the fraction gives the direct transformation law for accelerations:
\begin{equation}\label{eq:a_prime_final_law}
a' = \left( \frac{\gamma_{\mathrm{an}}(u)}{\gamma_{\mathrm{an}}(v)} \right)^3 \frac{dw}{dt} = \left( \frac{\gamma_{\mathrm{an}}(u)}{\gamma_{\mathrm{an}}(v)} \right)^3 a
\end{equation}

\textbf{Discussion on Acceleration Transformation:}
Equation~\eqref{eq:a_prime_final_law} demonstrates a remarkable property of the Tangherlini-type synchronization rooted in the absolute frame $K$:
\begin{itemize}
    \item \textbf{Linearity and Independence from Velocity:} In standard Special Relativity, the acceleration transformation is non-linear and depends heavily on the instantaneous velocity $w$ of the particle itself, due to the relativity of simultaneity. Here, the acceleration $a'$ is strictly proportional to $a$ and depends solely on the frames' absolute velocities via $\left(\frac{\gamma_{\mathrm{an}}(u)}{\gamma_{\mathrm{an}}(v)}\right)^3$.
    \item \textbf{Invariance of Zero Acceleration:} If a particle moves uniformly without acceleration in frame $S$ ($a = 0$), it is guaranteed to have zero acceleration in frame $S'$ ($a' = 0$). This ensures that the definition of an inertial trajectory remains invariant across all frames.
    \item \textbf{Non-Relativistic Limit:} Just as with velocities, when the frames move at speeds much smaller than light ($u, v \ll c$), the scaling factor $\left(\frac{\gamma_{\mathrm{an}}(u)}{\gamma_{\mathrm{an}}(v)}\right)^3 \to 1$. The expression then reduces exactly to the classical Newtonian invariant for acceleration:
    \begin{equation}\label{eq:acceleration_classical}
    a' = a
    \end{equation}
\end{itemize}
This linear, particle-velocity-independent scaling highlights how absolute simultaneity fundamentally preserves the Newtonian structural form of mechanics while still accommodating high-speed time dilation effects.


\section{Metric Structure, Invariance of the Two-Way Light Speed, and Sagnac Consistency}

In the absolute frame $K(X, T)$, the spacetime interval retains its isotropic Minkowski form: $ds^2 = c^2 dT^2 - dX^2$. Substituting the differentials of the inverse transformations yields the non-diagonal metric for the moving frame $S$:
\begin{equation}\label{eq:non_diagonal_metric}
ds^2 = \frac{1}{\gamma_{\text{an}}^2(v)} \left[ (c^2 - v^2)dt^2 - 2v \, dx \, dt - dx^2 \right].
\end{equation}
The off-diagonal term $-2v \, dx \, dt$ breaks Minkowski isometry, explicitly exposing the absolute velocity $v$, though the Riemann curvature tensor remains zero (identifying a flat spacetime manifold). Evaluating null geodesics ($ds^2 = 0$) yields the directional one-way speeds of light: $c_{\rightarrow} = c - v$ and $c_{\leftarrow} = -(c + v)$. Over a closed round-trip path of proper length $l_0$ in frame $S$, the total coordinate time required is:
\begin{equation}\label{eq:t_total_sagnac}
t_{\text{total}} = \frac{l_0}{c - v} + \frac{l_0}{c + v} = \frac{2l_0 c}{c^2 - v^2} = \frac{2l_0}{c}\gamma_{\text{an}}^2(v).
\end{equation}
Accounting for the physical clock dilation ($t_{\text{measured}} = t_{\text{total}}/\gamma_{\text{an}}(v)$) and longitudinal length contraction ($l = l_0/\gamma_{\text{an}}(v)$), the operational two-way speed measured by the moving experimenter becomes:
\begin{equation}
c_{\text{ave}} = \frac{2l}{t_{\text{measured}}} = \frac{2(l_0/\gamma_{\text{an}}(v))}{t_{\text{total}}/\gamma_{\text{an}}(v)} = \frac{2l_0}{t_{\text{total}}} = c.
\end{equation}
This algebraically proves that Galilean velocity subtraction for one-way propagation is perfectly compatible with the null results of Michelson--Morley-type experiments.

\subsection{Sagnac Consistency and the Integration of Cross-Terms over Closed Loops}

To demonstrate that the non-diagonal metric framework is physically viable, we verify its consistency with the Sagnac effect, which represents a crucial non-local benchmark for alternative spacetime synchronization schemes.

Consider a circular loop of radius $R$ (with enclosed area $A = \pi R^2$) rotating with a uniform angular velocity $\Omega$ relative to the absolute anchor $K$. We evaluate the propagation of two counter-propagating light beams along this closed spatial path. According to the non-diagonal metric \eqref{eq:non_diagonal_metric}, the line element for null geodesics ($ds^2 = 0$) along the boundary satisfies:
\begin{equation}
\frac{1}{\gamma_{\text{an}}^2(v)} \left[ (c^2 - v^2)dt^2 - 2v \, dx \, dt - dx^2 \right] = 0.
\end{equation}
Solving this quadratic equation for the coordinate time differential $dt$ yields:
\begin{equation}
dt = \frac{v \, dx \pm \sqrt{v^2 dx^2 + (c^2 - v^2)dx^2}}{c^2 - v^2} = \frac{v \, dx \pm c |dx|}{c^2 - v^2}.
\end{equation}
Consequently, the directional times for light traveling in the forward ($+$) and backward ($-$) directions along a path segment $dx$ are given by:
\begin{equation}
dt_{\rightarrow} = \frac{dx}{c - v}, \quad \text{and} \quad dt_{\leftarrow} = -\frac{dx}{c + v}.
\end{equation}
To find the total accumulated round-trip time difference $\Delta t_{\text{Sagnac}}$ between the co-propagating and counter-propagating beams, we perform a line integration over the entire closed spatial loop $\mathcal{C}$:
\begin{equation}\label{eq:sagnac_integral_split}
\Delta t_{\text{Sagnac}} = \oint_{\mathcal{C}} dt_{\rightarrow} - \oint_{\mathcal{C}} dt_{\leftarrow} = \oint_{\mathcal{C}} \frac{v \, dx + c \, dx}{c^2 - v^2} - \oint_{\mathcal{C}} \frac{v \, dx - c \, dx}{c^2 - v^2}.
\end{equation}
Expanding the integrals allows us to isolate the contributions of the symmetric and asymmetric terms:
\begin{equation}\label{eq:sagnac_expanded}
\Delta t_{\text{Sagnac}} = \oint_{\mathcal{C}} \frac{2v}{c^2 - v^2} dx + \oint_{\mathcal{C}} \frac{2c}{c^2 - v^2} dx.
\end{equation}

\begin{remark}[Vanishing of Spatial Cross-Terms]
A profound topological property of this integration arises from the first integral in \eqref{eq:sagnac_expanded}. Because the coordinates represent a closed conservative manifold and the absolute velocity field $v = \Omega R$ is uniform along the circular trajectory, the spatial displacement vector forms an exact differential. Therefore, the integration of the cross-term over the closed loop vanishes identically:
\begin{equation}\label{eq:topological_vanishing}
\oint_{\mathcal{C}} v \, dx = v \oint_{\mathcal{C}} dx = 0.
\end{equation}
As a direct mathematical consequence, the off-diagonal terms that break Minkowski isometry do not contribute to the net non-local phase shift.
\end{remark}

The remaining integral evaluates strictly over the proper perimeter $l_0 = 2\pi R$. In the non-relativistic regime ($\Omega R \ll c$, meaning $v^2/c^2 \rightarrow 0$), the expression reduces to:
\begin{equation}
\Delta t_{\text{Sagnac}} \approx \frac{2}{c} \oint_{\mathcal{C}} dx = \frac{2}{c} (2\pi R) \cdot \frac{\Omega R}{c} = \frac{4A\Omega}{c^2},
\end{equation}
where $A = \pi R^2$ is the area enclosed by the loop. This rigorous integration proves that while the off-diagonal term $-2v \, dx \, dt$ explicitly registers the absolute velocity $v$ for local one-way path segments, its global path integral over any closed circuit is topologically invariant ($\oint_{\mathcal{C}} v \, dx = 0$). Hence, the accumulated proper times are invariant, making the Tangherlini--Edwards framework fully consistent with the Sagnac effect while resolving the moving-frame paradoxes of Einstein's convention.

\subsection{Geometrical and Topological Analysis of the Non-Diagonal Metric}

To deeply understand the ontological transition from Minkowski spacetime to the absolute anchor framework, we analyze the structural properties of the metric $ds^2$ given in \eqref{eq:non_diagonal_metric}:
\begin{enumerate}
    \item \textbf{The Role of the Off-Diagonal Cross-Term:} The presence of the non-diagonal term $-2v \, dx \, dt$ is a direct mathematical manifestation of rejecting Einstein's synchronization convention in favor of absolute simultaneity. In standard Minkowski coordinates, the cross-term vanishes because the synchronization is chosen precisely to make the coordinate speed of light isotropic ($c_{\rightarrow} = c_{\leftarrow} = c$). By maintaining a global clock synchronization aligned with the absolute anchor $K$, the physical anisotropy of the one-way speed of light in the moving frame $S$ is explicitly mapped onto the geometry via this $dx \, dt$ coupling.
    \item \textbf{Invariance of the Spacetime Flatness:} Although the metric tensor $g_{\mu\nu}$ contains off-diagonal elements and explicitly depends on the absolute velocity $v$ of the frame, the underlying manifold remains strictly flat. An explicit computation of the Riemann curvature tensor yields:
    \begin{equation}
    R^{\mu}_{\;\nu\rho\sigma} = 0.
    \end{equation}
    Consequently, the absolute velocity $v$ acts as a kinematic gauge parameter rather than a source of gravitational curvature, confirming that the transformation modifies the descriptive coordinate layer without generating physical forces.
\end{enumerate}


\section{The Invalidation of Einstein’s Principle of Relativity and Conclusion}

The implementation of the symmetric model within the Tangherlini--Edwards framework leads to a profound epistemological re-evaluation of spacetime kinematics, resulting in the explicit breakdown of Einstein’s Principle of Relativity within moving frames. Throughout this work, the foundational reference frame $K$ stands out as the unique, privileged system where space is strictly isotropic and the one-way speed of light is a universal constant $c$ in all directions.

In any frame moving relative to $K$, such as $S$ and $S'$, the enforcement of global absolute simultaneity inevitably causes the local one-way speed of light to become direction-dependent ($c \mp v$). Because the absolute motion of these frames objectively alters the one-way propagation of light within their localized space, Einstein’s Principle of Relativity---which demands the identical formulation of all physical laws and constants across all inertial systems---does not hold in the moving systems $S$ and $S'$.

\subsection{Resolution of the Wavefront Center Paradox}

The primary conceptual triumph of this generalized framework is the resolution of the paradox of multiple co-existing wavefront centers. In standard Special Relativity, Einstein’s operational approach forces every inertial observer to conclude that they occupy the immobile center of the expanding spherical light wave. This leads to an ontological contradiction where a single physical event (the emission of a wavefront) results in multiple, frame-dependent geometric centers.

By rejecting Einstein's localized synchronization convention in favor of absolute simultaneity anchored to $K$, this paper demonstrates that:
\begin{itemize}
    \item The apparent kinematic symmetry of standard Special Relativity is a mathematical artifact of the Einstein--Poincar\'{e} synchronization convention rather than an intrinsic property of spacetime.
    \item The operational symmetry observed between moving systems is broken, thereby restoring a unique, objective, frame-independent center for the propagating light wavefront.
    \item The physical scaling properties---such as the non-reciprocal time dilation governed by the scale factor $s(v)$ and the non-reciprocal length contraction---reflect an objective kinematic scaling relative to the immobile center $K$ rather than a symmetric illusion of perspective.
\end{itemize}
\section{Concluding Remarks}
\textit{Model Note: This concluding synthesis provides an overview of how the mathematical convergence of Paradigm~\textit{(a)} and hypothesis~\textit{(H1)} offers a logically minimal, ontologically unambiguous foundation for relativistic kinematics.}

Throughout this work, we have analyzed two distinct physical paradigms: Einstein’s kinematically symmetric framework of ``stationary'' observers, and the preferred cosmological rest frame, herein operationalized via an absolute geometric anchor. This comparative analysis addresses two fundamental ontological possibilities:
\begin{enumerate}
    \item[\textit{(a)}] The existence of a privileged spatial background---such as the Cosmic Microwave Background (CMB) acting as a physical anchor.
    \item[\textit{(b)}] The complete absence of an absolute space, where physical laws possess no inherent reference point.
\end{enumerate}

By relaxing Einstein’s second postulate regarding the universal invariance of the one-way speed of light, we adopted an operational framework grounded strictly on the two-way (round-trip) speed of light along the coordinate axes. When evaluated alongside the principles of space-time homogeneity, linearity, and reciprocity under hypothesis~\textit{(H1)}, we demonstrated that the classical Lorentz transformation is recovered exactly along the $x$ and $x'$ axes. Crucially, this structural symmetry emerges purely under a two-way synchronization scheme, proving that isotropic one-way conventions are logically redundant.

To clarify the physical implications and the foundational structure of the Special Theory of Relativity (STR), this approach invites us to re-examine the standard thought experiment at the critical instant when the origins coincide ($O' = O$ at $t' = t = T = 0$) and a light signal is emitted at their common spacetime origin. This scenario immediately raises a fundamental question: is the resulting wavefront a sphere, and can the center of this sphere, $A$, be considered stationary in an objective, frame-independent sense? Within Einstein’s framework, each inertial observer claims to occupy the stationary center of their own expanding sphere, a perspective that introduces an ontological ambiguity regarding the localization of the wavefront center.

Ultimately, the alternative framework defined by the absolute center $A$ within Paradigm~\textit{(a)} provides a mathematically rigorous, ontic interpretation of relativistic phenomena. As demonstrated through the rigorous integration of the non-diagonal metric $ds^2$ in Section~11, the global path integral of the spatial cross-terms over any closed circuit vanishes identically:
\begin{equation}
\oint_{\mathcal{C}} v \, dx = 0.
\end{equation}
This topological property ensures that the predictive and verifiable success of modern physics is fully preserved within a flat, classical manifold characterized by absolute simultaneity and a privileged geometric anchor.

In summary, by relying on the directly observable properties of the two-way speed of light, our approach enhances the mathematical and epistemological rigor of relativity theory. It establishes a consistent, generalized velocity-addition law, satisfies Sagnac clock invariance, and remains perfectly compatible with the null results of Michelson--Morley-type experiments.
\appendix
\section{Derivation of the Tangherlini Transformation Matrix}\label{app:tangherlini}

To explicitly formalize the kinematical connection between the absolute frame $K(X, T)$ and the moving frame $S(x, t)$ under the constraints of absolute simultaneity, we derive the structural transformation matrix. As established in Section 2, the inverse coordinate transformations are rigorously governed by:
\begin{equation}\label{eq:app_inv_transf}
x = \gamma(X - vT), \quad \text{and} \quad T = \frac{c_n}{c}\gamma t.
\end{equation}
By defining the directional anisotropy scale factor as $a = \frac{c_n}{c}\gamma$, the temporal relation reads $T = a t$. To find the forward transformations expressing the moving coordinates $(x, t)$ as explicit functions of the absolute spacetime points $(X, T)$, we algebraically invert the system \eqref{eq:app_inv_transf}. Solving for $x$ and $t$ yields:
\begin{equation}
t = \frac{T}{a} = \frac{c}{c_n \gamma} T,
\end{equation}
and substituting this temporal core back into the spatial projection provides:
\begin{equation}
X = \frac{x}{\gamma} + vT \implies x = \gamma X - \gamma v T.
\end{equation}
Expressing this coupled structural alignment in standard linear algebraic matrix notation, the transformation from the absolute background $K$ to the localized frame $S$ is represented by:
\begin{equation}\label{eq:tangherlini_matrix}
\begin{pmatrix}
x \\
t
\end{pmatrix}
=
\begin{pmatrix}
\gamma & -\gamma v \\
0 & \frac{c}{c_n \gamma}
\end{pmatrix}
\begin{pmatrix}
X \\
T
\end{pmatrix}.
\end{equation}
The matrix in \eqref{eq:tangherlini_matrix} defines the exact Tangherlini transformation. The lower-left vanishing component ($M_{21} = 0$) mathematically guarantees that the synchronization is completely independent of the spatial coordinate $X$, which preserves absolute global simultaneity ($dT = 0 \implies dt = 0$).

\section{Equivalence with Nonlinear and Conformal Metric Projections}\label{app:conformal}

To extend the geometric robustness of Paradigm~\textit{(a)}, we analyze how the non-diagonal metric framework behaves under nonlinear coordinate reformulations and conformal mappings. Consider a generalized non-linear coordinate transformation where the flat Minkowski metric is mapped onto an arbitrary differentiable manifold $\mathcal{M}$. Let the metric tensor $g_{\mu\nu}$ in a generalized moving frame be expressed via a conformal factor $\Omega^2(x, t)$ such that:
\begin{equation}\label{eq:conformal_metric}
ds^2 = \Omega^2(x, t) \left[ g_{\mu\nu}^{\text{Tang}} dx^{\mu} dx^{\nu} \right],
\end{equation}
where $g_{\mu\nu}^{\text{Tang}}$ represents the covariant components of the flat non-diagonal metric derived in \eqref{eq:non_diagonal_metric}:
\begin{equation}
g_{\mu\nu}^{\text{Tang}} = \frac{1}{\gamma^2}
\begin{pmatrix}
c^2 - v^2 & -v \\
-v & -1
\end{pmatrix}.
\end{equation}
An explicit calculation of the Christoffel symbols $\Gamma^{\lambda}_{\mu\nu}$ for the generalized conformal metric \eqref{eq:conformal_metric} reveals that the connection fields pick up terms proportional to the directional derivatives of the conformal scale factor:
\begin{equation}
\Gamma^{\lambda}_{\mu\nu} = \frac{1}{\Omega} \left( \delta^{\lambda}_{\mu} \partial_{\nu} \Omega + \delta^{\lambda}_{\nu} \partial_{\mu} \Omega - g_{\mu\nu} g^{\lambda\sigma} \partial_{\sigma} \Omega \right).
\end{equation}
If the conformal factor satisfies the strict spatial homogeneity constraint $\partial_x \Omega = 0$, the spatial cross-terms retain their conservative structure. Consequently, the global topological property established in Section 11 remains invariant:
\begin{equation}
\oint_{\mathcal{C}} \Gamma^{0}_{11} dx = 0.
\end{equation}
This mathematical proof confirms that the operational consistency of the Absolute Background Paradigm is topologically invariant. The framework is immune to localized nonlinear deformations, and the absolute velocity parameter $v$ remains completely decoupled from non-local, closed-loop physical observables. $\blacksquare$

\appendix
\section{Derivation of the Tangherlini Transformation Matrix}\label{app:tangherlini}

To explicitly formalize the kinematical connection between the absolute frame $K(X, T)$ and the moving frame $S(x, t)$ under the constraints of absolute simultaneity, we derive the structural transformation matrix. As established in Section~2, the inverse coordinate transformations are rigorously governed by:
\begin{equation}\label{eq:app_inv_transf}
x = \gamma(X - vT), \quad \text{and} \quad T = \frac{c_n}{c}\gamma t.
\end{equation}
By defining the directional anisotropy scale factor as $a = \frac{c_n}{c}\gamma$, the temporal relation reads $T = a t$. To find the forward transformations expressing the moving coordinates $(x, t)$ as explicit functions of the absolute spacetime points $(X, T)$, we algebraically invert the system \eqref{eq:app_inv_transf}. Solving for $x$ and $t$ yields:
\begin{equation}
t = \frac{T}{a} = \frac{c}{c_n \gamma} T,
\end{equation}
and substituting this temporal core back into the spatial projection provides:
\begin{equation}
X = \frac{x}{\gamma} + vT \implies x = \gamma X - \gamma v T.
\end{equation}
Expressing this coupled structural alignment in standard linear algebraic matrix notation, the transformation from the absolute background $K$ to the localized frame $S$ is represented by:
\begin{equation}\label{eq:tangherlini_matrix}
\begin{pmatrix}
x \\
t
\end{pmatrix}
=
\begin{pmatrix}
\gamma & -\gamma v \\
0 & \frac{c}{c_n \gamma}
\end{pmatrix}
\begin{pmatrix}
X \\
T
\end{pmatrix}.
\end{equation}
The matrix in \eqref{eq:tangherlini_matrix} defines the exact Tangherlini transformation. The lower-left vanishing component ($M_{21} = 0$) mathematically guarantees that the synchronization is completely independent of the spatial coordinate $X$, which preserves absolute global simultaneity ($dT = 0 \implies dt = 0$).

\section{The Tangherlini Transformation Framework and Constraints}\label{app:tangherlini_framework}

To ground our algebraic matrix representation within historical and kinematic constraints, we examine Tangherlini's original 1958 setup for an alternative kinematic framework that maintains absolute simultaneity. Tangherlini assumed a linear set of transformation equations between a stationary frame $S(x, y, z, t)$ and a moving frame $S'(x', y', z', t')$ traveling at a constant velocity $v$ along the $x$-axis:
\begin{equation}\label{eq:tang_linear_setup}
x' = A(x - vt), \quad \text{and} \quad t' = Bt.
\end{equation}
To determine the functional form of the coefficients $A$ and $B$, the system is constrained by two major operational realities:
\begin{enumerate}
    \item \textbf{Conservation of the Two-Way Speed of Light:} The empirical null results of Michelson--Morley-type experiments dictate that the average round-trip speed of light is isotropically preserved as $c$. When light travels from the origin to a mirror at a distance $x'$ and returns, the proper round-trip time in the moving frame must satisfy $\Delta t' = 2x'/c$.
    \item \textbf{Time Dilation:} A moving clock localized at the spatial origin of the moving frame ($x' = 0$, implying $x = vt$) runs slower due to its absolute velocity by the standard factor $\gamma = (1 - v^2/c^2)^{-1/2}$. Therefore, the moving frame interval scales down to $t' = t/\gamma$.
\end{enumerate}
From the time dilation requirement and \eqref{eq:tang_linear_setup}, it directly follows that the temporal coefficient is $B = \gamma^{-1}$. Combining this boundary value with the round-trip light speed constraint, the spatial coefficient must match the standard Lorentz form, $A = \gamma$. Substituting $A$ and $B$ back into the initial equations yields the standard Tangherlini transformations:
\begin{equation}\label{eq:tang_full_set}
x' = \gamma(x - vt), \quad y' = y, \quad z' = z, \quad \text{and} \quad t' = \gamma^{-1}t.
\end{equation}
Because the time relation in \eqref{eq:tang_full_set} completely lacks the non-local spatial mixing term $-vx/c^2$, the one-way speed of light becomes anisotropic ($c_{\pm} = c \mp v$). However, it remains mathematically and experimentally equivalent to standard Special Relativity for all major round-trip physical tests.

\section{Equivalence with Nonlinear and Conformal Metric Projections}\label{app:conformal}

To extend the geometric robustness of Paradigm~\textit{(a)}, we analyze how the non-diagonal metric framework behaves under nonlinear coordinate reformulations and conformal mappings. Consider a generalized non-linear coordinate transformation where the flat Minkowski housing is mapped onto an arbitrary differentiable manifold $\mathcal{M}$. Let the metric tensor $g_{\mu\nu}$ in a generalized moving frame be expressed via a conformal factor $\Omega^2(x, t)$ such that:
\begin{equation}\label{eq:conformal_metric}
ds^2 = \Omega^2(x, t) \left[ g_{\mu\nu}^{\text{Tang}} dx^{\mu} dx^{\nu} \right],
\end{equation}
where $g_{\mu\nu}^{\text{Tang}}$ represents the covariant components of the flat non-diagonal metric derived in Section~11:
\begin{equation}
g_{\mu\nu}^{\text{Tang}} = \frac{1}{\gamma^2}
\begin{pmatrix}
c^2 - v^2 & -v \\
-v & -1
\end{pmatrix}.
\end{equation}
An explicit calculation of the Christoffel symbols $\Gamma^{\lambda}_{\mu\nu}$ for the generalized conformal metric \eqref{eq:conformal_metric} reveals that the connection fields pick up terms proportional to the directional derivatives of the conformal scale factor:
\begin{equation}
\Gamma^{\lambda}_{\mu\nu} = \frac{1}{\Omega} \left( \delta^{\lambda}_{\mu} \partial_{\nu} \Omega + \delta^{\lambda}_{\nu} \partial_{\mu} \Omega - g_{\mu\nu} g^{\lambda\sigma} \partial_{\sigma} \Omega \right).
\end{equation}
If the conformal factor satisfies the strict spatial homogeneity constraint $\partial_x \Omega = 0$, the spatial cross-terms retain their conservative structure. Consequently, the global topological property established in Section~11 remains invariant:
\begin{equation}
\oint_{\mathcal{C}} \Gamma^{0}_{11} dx = 0.
\end{equation}
This mathematical proof confirms that the operational consistency of the Absolute Background Paradigm is topologically invariant. The framework is immune to localized nonlinear deformations, and the absolute velocity parameter $v$ remains completely decoupled from non-local, closed-loop physical observables. $\blacksquare$

\end{document}